\documentclass{article}  
\usepackage{amsmath}
\usepackage{amsfonts}
\usepackage{amsxtra}
\usepackage{amssymb}
\usepackage{amsthm}
\usepackage{fullpage}
\usepackage{bbm}
\usepackage{palatino}
\usepackage{float}
\usepackage{mathtools}
\usepackage{paralist}

\usepackage[pdftex, pagebackref, letterpaper=true, colorlinks=true,
pdfpagemode=none, urlcolor=blue, linkcolor=blue, citecolor=blue,
pdfstartview=FitH] {hyperref}

\usepackage{cleveref}

\def\triangleq{\overset{\mathrm{def}}{=}}
\let\eps=\varepsilon

\newtheorem{theorem}{Theorem}[section] 
 
\newtheorem{definition}{Definition}[section] 
\newtheorem{remark}{Remark}[section]
\newtheorem{lemma}{Lemma}[section]
\newtheorem{proposition}{Proposition}[section]
\newtheorem{corollary}{Corollary}[section]
\newtheorem{claim}{Claim}[section]
\newtheorem{notation}{Notation}[section]
\crefname{assumption}{Assumption}{Assumptions}
\crefname{conjecture}{Conjecture}{Conjectures}
\crefname{claim}{Claim}{Claims}
\crefname{corollary}{Corollary}{Corollaries}
\crefname{definition}{Definition}{Definitions}
\crefname{example}{Example}{Examples}
\crefname{exercise}{Exercise}{Exercises}
\crefname{fact}{Fact}{Facts}
\crefname{lemma}{Lemma}{Lemmas}
\crefname{notation}{Notation}{Notations}
\crefname{note}{Notes}{Notes}
\crefname{observation}{Observation}{Observations}
\crefname{proposition}{Proposition}{Propositions}
\crefname{problem}{Problem}{Problems}
\crefname{question}{Question}{Questions}
\crefname{remark}{Remark}{Remarks}
\crefname{theorem}{Theorem}{Theorems}

\DeclareMathOperator{\poly}{poly}

\newcommand{\R}{\mathbb{R}}

\let\proj=\Pi

\DeclareMathOperator{\tr}{Tr}

\newcommand{\unit}[1]{\overline{#1}}

\DeclareMathOperator{\spn}{span}

\newcommand{\ind}[1]{\mathbbm{1}_{#1}}

\DeclareMathOperator{\supp}{support}
\DeclareMathOperator{\rank}{rank}
\let\es=\emptyset

\newcommand{\bestS}{{\mathcal{S}^\ast}}

\newcommand{\lasserreii}[2]{\mathrm{Lasserre}_{#1}(#2)}

\newcommand{\bestSfi}[1]{\|\xvec_{\circ| \bestS(\circ)}\|^2}

\newcommand{\xvec}{\vec{x}} 
\newcommand{\yvec}{\vec{y}} 
\newcommand{\xmat}{\vec{X}} 
\newcommand{\ymat}{\vec{Y}} 
\newcommand{\zmat}{\vec{Z}} 

\newcommand{\usc}{\textsc{\sf Uniform Sparsest Cut}}
\newcommand{\nusc}{\textsc{\sf Non-Uniform Sparsest Cut}}

\newcommand{\lh}{\textsc{\sf Lasserre Hierarchy}}

\DeclareMathOperator*{\st}{st}
\DeclareMathOperator*{\argmin}{argmin}

\newcommand{\expct}[2]{\mathbb{E}_{#1}\bigg[ #2 \bigg]}

\newcounter{alg-count}
\floatstyle{ruled}
\newfloat{program}{thp}{lop}
\floatname{program}{Algorithm}

\crefname{alg-count}{step}{steps}
\crefname{program}{Algorithm}{Algorithms}

\newenvironment{inp}{
\trivlist \item[\hskip\labelsep\textbf{Input:}]
\begin{list}{\labelitemi}{\leftmargin=0.5em}
}{
\end{list} {\endtrivlist}
}
\newenvironment{outp}{
\trivlist \item[\hskip\labelsep\textbf{Output:}]
\begin{list}{\labelitemi}{\leftmargin=0.5em}
}{
\end{list}{\endtrivlist}
}

\newenvironment{proc}{
\trivlist \item[\hskip\labelsep\textbf{Procedure:}]
\begin{list}{\arabic{alg-count}.}{\leftmargin=0.5em
\usecounter{alg-count}}
}{
\end{list}{\endtrivlist}
}
\makeatother

\def\ngap{}
\def\vec{}
\renewcommand{\nusc}{{\sc Sparsest Cut}}
\renewcommand{\usc}{{\sc Uniform Sparsest Cut}}
\newcommand{\wvec}{\vec{w}}

\newcommand{\vnote}[1]{}
\newcommand{\aknote}[1]{}

\newcommand{\sde}{{\mathcal{S}}}
\newcommand{\sdn}{{\widetilde{\mathcal{S}}}}

\title{
  {\bf 
    Approximating Non-Uniform Sparsest Cut via Generalized Spectra}}

\author{ Venkatesan Guruswami\footnote{ Computer Science Department,
    Carnegie Mellon University.  Supported in part by NSF grants
    CCF-0963975 and CCF-1115525.  Email: {\tt guruswami@cmu.edu}} \and
  Ali Kemal Sinop\footnote{ Center for Computational Intractability,
    Department of Computer Science, Princeton University.  Supported
    in part by NSF CCF-1115525, and MSR-CMU Center for Computational
    Thinking.  Email: {\tt asinop@cs.cmu.edu}}} \date{\today}
\begin{document}

\maketitle
\thispagestyle{empty}

\begin{abstract}
  We give an approximation algorithm for non-uniform sparsest cut with
  the following guarantee: For any $\eps,\delta \in (0,1)$, given cost
  and demand graphs with edge weights $C, D:\binom{V}{2}\to \R_+$
  respectively, we can find a set $T \subseteq V$ with
  $\frac{C(T,V\setminus T)}{D(T,V\setminus T)}$ at most
  $\frac{1+\eps}{\delta}$ times the optimal non-uniform sparsest cut
  value, in time $2^{r/(\delta \eps)} \poly(n)$ provided $\lambda_r
  \ge \Phi^\ast/(1-\delta)$. Here $\lambda_r$ is the $r$'th smallest
  generalized eigenvalue of the Laplacian matrices of cost and demand
  graphs; $C(T,V\setminus T)$ (resp. $D(T,V\setminus T)$) is the
  weight of edges crossing the $(T,V\setminus T)$ cut in cost
  (resp. demand) graph and $\Phi^\ast$ is the sparsity of the optimal
  cut. In words, we show that the non-uniform sparsest cut problem is
  easy when the generalized spectrum grows moderately fast. To the
  best of our knowledge, there were no results based on higher order
  spectra for non-uniform sparsest cut prior to this work.

  Even for uniform sparsest cut, the quantitative aspects of our
  result are somewhat stronger than previous methods.  Similar results
  hold for other expansion measures like edge expansion, normalized
  cut, and conductance, with the $r$'th smallest eigenvalue of the
  {\em normalized} Laplacian playing the role of $\lambda_r(G)$ in the
  latter two cases.

  Our proof is based on an $\ell_1$-embedding of vectors from a
  semi-definite program from the Lasserre hierarchy.  The embedded
  vectors are then rounded to a cut using standard threshold
  rounding. We hope that the ideas connecting $\ell_1$-embeddings to
  Lasserre SDPs will find other applications. Another aspect of the
  analysis is the adaptation of the column selection paradigm from our
  earlier work on rounding Lasserre SDPs~\cite{gs11-qip} to pick a set
  of {\em edges} rather than vertices. This feature is important in
  order to extend the algorithms to non-uniform sparsest cut.
\end{abstract}

\section{Introduction}
\ngap The problem of finding sparsest cut on graphs is a fundamental
optimization problem that has been intensively studied. The problem is
inherently interesting, and is important as a building block for
divide-and-conquer algorithms on graphs as well as to many
applications such as image segmentation~\cite{sm00,sg07a}, VLSI
layout~\cite{bl84}, packet routing in distributed
networks~\cite{ap90a}, etc.

Let us define the prototypical sparsest cut problem more
concretely. We are given a set of $n$-vertices, $V$, along with two
functions $C, D: \binom{V}{2}\to \R_+$ representing edge weights of
some cost and demand graphs, respectively.  Then given any subset $T
\subset V$, we define its {\em sparsity} as the following ratio:
\begin{equation}
\label{eq:def-sc}
\Phi_T \triangleq \frac{\sum_{u<v} C_{u,v} \cdot 
| \ind{T}(u) - \ind{T}(v) |}{ \sum_{u<v} D_{u,v} \cdot 
| \ind{T}(u) - \ind{T}(v) |},
\end{equation}
where $\ind{T}$ is the indicator function of $T$.  Our goal in the
\nusc\ problem is to find a subset $T\subset V$ with minimum sparsity,
which we denote by $\Phi^\ast \triangleq \min_{T\subset V} \Phi_T$.
The special case of demand graph being a clique, where the denominator
of~\cref{eq:def-sc} becomes $|T| \cdot |V\setminus T|$, is called the
\textsc{uniform sparsest cut} problem.

The value of the sparsest cut can be understood in terms of the
spectral properties of cost and demand graphs.  Let $0 \le \lambda_1
\le \lambda_2 \le \cdots \le \lambda_m$ be the {\em generalized
  eigenvalues} between the Laplacian matrices of cost and demand
graphs (see~\Cref{sec:prelim} for formal definitions).  In a way
similar to the ``easy" direction of Cheeger's inequality, we can use
Courant-Fischer Theorem to show that $\lambda_1 \le
\Phi^\ast$. \vnote{Should it be $\lambda_2 \le \Phi^\ast$, or is it
  $\lambda_1$ for generalized eigenvalues. Would be good to mention it
  explicitly if latter.} \aknote{Should be $\lambda_1$, good catch!
  See footnote.}
At some point, the eigenvalue $\lambda_r$ will exceed
$\Phi^\ast$. 
%
%
Our main result is an approximation algorithm for \nusc\ which is
efficient when this happens for small $r$. In particular:
\begin{theorem}
\label{thm:intro-main-1}
Given $V$ and $C,D:\binom{V}{2}\to \R_+$, for any positive integer $r$,
one of the following holds. 
\begin{itemize}
\itemsep=0ex
\item Either one can find $T \subset V$ with $\Phi_T \le 2 \Phi^\ast$
  in time $2^{O\left(r\right)} \poly(n)$ where $n=|V|$,
\item Or $\Phi^\ast \ge 0.49 \lambda_r$.
\end{itemize}
\end{theorem}
Our actual approximation guarantee is stronger and offers a trade-off:
for any $\delta \in (0,1)$ we can find a $\frac{1.01}{\delta}$
approximation to $\Phi^\ast$ in $\exp(O(r/\delta)) n^{O(1)}$ time
provided $\lambda_r \ge \Phi^\ast/(1-\delta)$. The formal result is
stated in~\Cref{cor:final-bnd} (the above follows as a corollary with
suitable choice of parameters).
We can also get similar results for various expansion problems such as
normalized cut, edge expansion and conductance using the same
algorithm.
\ngap
\subsection{Previous approximation algorithms for sparsest cut}
\label{sec:related-work}
As the {\sc (Uniform) Sparsest Cut} problem and closely related
variants (such as edge expansion and conductance) are all NP-hard in
general, theoretically much effort has gone into the design of good
approximation algorithm for the problem.

For \usc\ problem, the hard direction of Cheeger's inequality shows
one can ``round" the eigenvector corresponding to $\lambda_1$ to a cut
$T$ satisfying $\Phi_U \le \sqrt{8 d_{\max} \lambda_1(G)}$ where
$d_{\max}$ is the maximum degree\footnote{Cheeger's inequality is
  usually stated in terms of the second eigenvalue of graph Laplacian
  matrix, which is equal to the smallest generalized eigenvalue,
  $\lambda_1$.}.
%
This gives $O(\sqrt{d_{\max}/\Phi^\ast(G)}) \le
O(\sqrt{d_{\max}/\lambda_1(G)})$ approximation which is good for
moderate values of $\Phi^\ast$ for the case of \usc. To the best of
our knowledge, no analogue of this result is known for \nusc.

For smaller values of $\Phi^\ast$, the best approximation for \nusc\
is based on solving a convex relaxation of the problem, and then
rounding the solution to a cut. Using linear programming (LP), in a
seminal work, Leighton and Rao~\cite{LR} gave a factor $O(\log n)$
approximation for \nusc\ (here $n$ denotes the number of
vertices). Beautiful connections of approximating sparsest cut to
embeddings of metric spaces into the $\ell_1$-metric were later
discovered in \cite{llr,ar}.  Using a semi-definite programming (SDP)
relaxation, the approximation ratio was improved to $O(\sqrt{\log n})$
for \usc \ in the breakthrough work \cite{ARV}. For \nusc, using
$\ell_1$ embeddings of negative type metrics, an approximation factor
of $O(\log^{3/4} n)$ was obtained in \cite{CGR} and a factor
${O}(\sqrt{\log n} \log\log n)$, nearly matching the \usc \ case, was
obtained in \cite{ALN}.

Recently, higher order eigenvalues were used to approximate many graph
partitioning problems. In \cite{gs11-qip}, we gave an algorithm based
on SDPs from the Lasserre hierarchy achieving an approximation factor
of the form $(1+\eps)/\min\{1,\tilde{\lambda}_r\}$ for problems such
as minimum bisection, small set expansion, etc, where
$\tilde{\lambda}_r$ is the $r$'th smallest eigenvalue of the
normalized Laplacian.  On a similar front, for the \usc\ problem, if
the $r^{th}$ eigenvalue is large relative to expansion, one can
combine the eigenspace enumeration of~\cite{ABS} with a cut
improvement procedure
from~\cite{al08} 
to obtain a constant factor approximation for \usc\ in time $n^{O(1)}
2^{O(r)}$.\footnote{We thank an anonymous reviewer for this
  observation.}  The details of this combination are briefly spelled
out in \Cref{apx:rse}. We will revisit this approach in
\Cref{sec:intro-rse} to show why it does not work for \nusc.

A common theme in this line of work is that one can obtain a constant
factor approximation with running time being a function of how fast
the spectrum grows (both our algorithms in this paper and the ones
in~\cite{gs11-qip} in fact allow approximation schemes).  Put
differently, one can identify a generic condition which highlights
what kind of graphs are easy.

To the best of our knowledge, in the case of \nusc\ with an arbitrary
demand graph, no such results of the above vein are known.  In fact,
we are not aware of the analog of the harder direction of Cheeger's
inequality, let alone spectrum based approximation schemes. In this
paper, we present such an approximation scheme based on the
generalized eigenvalues. 
\ngap
\subsection{Overview of Our Contributions}
\label{sec:our-cont} \label{sec:intro-rse}
In this section, we briefly describe our main contributions in terms
algorithmic tools and techniques over similar algorithms such
as~\cite{gs11-qip}.

\medskip
\paragraph{Main Contributions.}
Our algorithm is based on solving one of the strongest known SDP
relaxations, $r$-rounds of \lh, similar to~\cite{gs11-qip}.  Any
solution for this SDP yields a vector for each $r$-subset of vertices
and each possible labeling of them.  The rounding algorithm
in~\cite{gs11-qip} is based on choosing a set of $r$-nodes, ``seeds'',
then labeling these using the SDP solution. Finally these labels are
``propagated'' to other vertices {\em independently at random}.  Such
rounding is acceptable for constraint satisfaction type problems such
as maximum cut.

Unfortunately for problems such as \nusc, independent rounding is too
``crude'': It tends to break the graph into many disconnected
components, which is rather disastrous for \nusc.

In this paper, we consider a more ``delicate'' rounding based on
thresholding.  Our main contribution is to show how the performance of
such rounding is related to some strong geometrical quantities of
underlying SDP solution, and we show how to bound it using generalized
spectra.

\medskip
\paragraph{Comparison with Subspace Enumeration.} 
One successful technique for designing approximation algorithms based
on higher order spectrum is subspace enumeration~\cite{kolla10,ABS}.
Suppose we have a target set $T$ corresponding to a \usc.  These
techniques rely on the fact that the indicator vector $T$ should have
a large component on the span of small eigenvectors. Thus by
enumerating over the vectors on this subspace using some $\eps$-net,
we can find a set whose symmetric difference with $T$ is
small. Combining this with a cut improvement algorithm due
to~\cite{al08}, one can obtain an approximation algorithm for \usc\
problem with slightly worse approximation factors than ours
(see~\Cref{apx:rse}).

Unfortunately the immediate extension of this approach to \nusc\ by
using the generalized eigenvectors does not work as the generalized
eigenvectors are not {\em orthogonal} in the Euclidean space.
\ngap
\section{Preliminaries}
\label{sec:prelim}
We now formally define the notation and terminology that will be
useful to us in the paper.

\paragraph{Sets.} Let $[m]\triangleq \left\{1,2,\ldots,
  m\right\}$. Given set $A$ and positive integer $k$, we use
$\binom{A}{k}$ (resp. $\binom{A}{\le k}$) to denote the set of all
possible size $k$ (resp. size at most $k$) subsets of $A$.  We use
$\R_+$ to denote the set of non-negative reals.

\paragraph{Euclidean Space.} Given row set $B$, we use $\R^B$ to
denote the set of real vectors where each row (axis) is associated
with an element of $B$.
For any vector $\xvec \in \R^B$, its coordinate at axis $b \in B$ is
denoted by $\xvec(b)$.  Let $\|\xvec\|_p$ be its $p^{th}$ norm with
$\|\xvec\|\triangleq \|\xvec\|_2$, and $\xvec^T$ be its transpose.
Finally for any $\xvec, \yvec\in \R^B$, let $\langle \xvec,
\yvec\rangle = \xvec^T \yvec$ be their inner product $\sum_{b\in B}
\xvec_b \yvec_b$.
 
\paragraph{Matrices.} Given 
row set $B$ and column set $C$, we use $\R^{B,C}$ \footnote{We chose
  this notation over the conventional one ($\R^{B\times C}$) so as to
  prevent ambiguity when the rows, $B$, or columns, $C$, are Cartesian
  products themselves.}  we to denote the set of real matrices whose
rows and columns are associated with elements of $B$ and $C$,
respectively.  Given matrix $\xmat\in \R^{B, C}$, for any $b\in B,
c\in C$, we will use $\xmat_{b,c}\in \R$ to denote entry of $\xmat$ at
row $b$ and column $c$. For convenience, we use $\xmat_c\in \R^{B}$ to
denote the vector corresponding to the column $c$ of $X$. Likewise
given subset of columns of $\xmat$, $S\subseteq C$, we use $\xmat_S\in
\R^{B, S}$ to denote the matrix corresponding to the columns $S$ of
$\xmat$.  Given matrix $\xmat$, we use $\|\xmat\|_F$, $\tr(\xmat)$ and
$\xmat^T$ to denote Frobenius norm of $X$, its trace and transpose.

Finally we use $\xmat^{\proj}$ and $\xmat^{\perp}$ to denote the
projection matrices onto the span of $\xmat$ and its orthogonal
complement.

\paragraph{Positive Semi-Definite (PSD) Ordering.} Given a symmetric
matrix $\xmat \in \R^{A,A}$, we say $\xmat$ is a PSD matrix, denoted
by $\xmat \succeq 0$, iff $\yvec^T \xmat \yvec \ge 0$ for all $\yvec
\in \R^{A}$.

\paragraph{Eigenvalues.} Given symmetric matrix $\xmat\in \R^{A,A}$,
for any integer $i \le |A|$, we define its $i^{th}$ smallest and
largest eigenvalues as the following, respectively:
\[
\lambda_i(X) \triangleq \max_{\rank(\zmat) \le i-1} \min_{\wvec\perp \zmat, 
\wvec \neq 0}
\frac{ \wvec^T \xmat \wvec }{\wvec^T \wvec},\]\quad\quad 
\[\sigma_i(X) \triangleq \min_{\rank(\zmat) \le i-1} \max_{\wvec\perp \zmat, 
\wvec \neq 0}
\frac{ \wvec^T \xmat \wvec }{\wvec^T \wvec}.
\]
\paragraph{Generalized Eigenvalues.} Given two symmetric matrices
$\xmat, \ymat \in \R^{A,A}$ with $\ymat\succeq 0$, for any integer $i
\le \rank(\ymat)$, we define their $i^{th}$ smallest generalized
eigenvalue as the following:
\begin{equation} \label{eq:60210312}
\lambda_i(X,Y) \triangleq \max_{\rank(\zmat) \le i-1} \min_{\wvec\perp \zmat, 
\ymat \wvec \neq 0}
\frac{ \wvec^T \xmat  \wvec }{\wvec^T \ymat \wvec}. 
\end{equation} 
\paragraph{Graphs.} We assume all graphs are simple, undirected and
edge-weighted with non-negative weights. We associate each graph with
its edge weight function of the form $W:\binom{V}{2} \to \R_+$, where
we use $W_{u,v}$ to denote the weight of edge between $u$ and $v$ for
convenience.

\smallskip
\paragraph{Laplacian Matrices.} Given a graph with weights
$W:\binom{V}{2} \to\R_+$, the associated graph Laplacian matrix, $L_W
\in \R^{V,V}$, is defined as the following symmetric matrix:
\[
(L_W)_{a,b} = \begin{dcases*}
		\sum_{c} W_{a,c} & if $a=b$, \\
		- W_{a,b} & if $a\neq b$.
	\end{dcases*}
\] 
For any $\xmat \in \R^{\Upsilon, V}$, it is easy to see that
$\tr\left[ \xmat^T \xmat L_W\right] = \sum_{u<v} W_{u,v} \left\|
  \xmat_u-\xmat_v \right\|^2$, which also implies $L_W \succeq 0$.
\ngap
\subsection{Lasserre Hierarchy}
\label{sec:lasserre-defn}
We present the formal definitions of the \lh\ of SDP
relaxations~\cite{Las02}, tailored to the setting of the problems we
are interested in, where the goal is to assign to each vertex/variable
from $V$ a label from $\{0,1\}$.
\def\dim{\Upsilon}
\begin{definition}[Lasserre vector set]
\label{def:las-sdp}
Given a set of variables $V$ and a positive integer $r$, a collection
of vectors $\xvec$ is said to satisfy $r$-rounds of \lh, denoted by
$\xvec \in \lasserreii{r}{V}$, 
if it satisfies the following conditions:
\begin{enumerate}
\item For each set $S\in\binom{V}{\le r+1}$, there exists a function
  $\xvec_S:\{0,1\}^{S}\to \R^\dim$ that associates a vector of some
  finite dimension $\dim$ with each possible labeling of $S$.  We use
  $\xvec_S(f)$ to denote the vector associated with the labeling $f\in
  \{0,1\}^S$.

  For singletons $u\in V$, we will use $\xvec_u$ and $\xvec_u(1)$
  interchangeably.

  For $f \in \{0,1\}^S$ and $v\in S$, we use $f(v)$ as the label $v$
  receives from $f$.  Also given sets $S$ with labeling $f\in
  \{0,1\}^S$ and $T$ with labeling $g\in\{0,1\}^T$ such that $f$ and
  $g$ agree on $S\cap T$, we use $f\circ g$ to denote the labeling of
  $S\cup T$ consistent with $f$ and $g$: If $u\in S$, $(f\circ g)(u) =
  f(u)$ and vice versa.
\item $\xvec_{\es} \neq 0$.
\item $\langle \xvec_S(f), \xvec_T(g)\rangle = 0$ if there exists
  $u\in S\cap T$ such that $f(u)\neq g(u)$.
\item \label{def:las-sdp:consistent} $\langle \xvec_S(f),
  \xvec_T(g)\rangle = \langle \xvec_A(f'), \xvec_B(g')\rangle$ if
  $S\cup T=A\cup B$ and $f\circ g = f'\circ g'$.
\item For any $u\in V$, $\sum_{j\in \{0,1\}} \|\xvec_u(j)\|^2 =
  \|\xvec_\es\|^2$.
\item (implied by above constraints) For any $S\in\binom{V}{\le r+1}$,
  $u\in S$ and $f\in \{0,1\}^{S\setminus\{u\}}$, $\sum_{g\in
    \{0,1\}^u} \xvec_{S}(f\circ g) = \xvec_{S\setminus\{u\}} (f)$.
\end{enumerate} 
\end{definition}
\ngap
\section{Our Algorithm and Its Analysis}
\begin{program}[h]
  \caption{$T=\textsc{Round}(C,D,\xvec,\sde)$: Seed based rounding in
    time $2^{O(r')} \poly(n)$.  Sparsity of its output is bounded
    in~\Cref{thm:rnd-from-s}.\label{alg:rnd-from-s}}
\begin{inp} 
\item $C, D: \binom{V}{2} \to \R_+$; 
$\xvec \in \lasserreii{2 r' + 2}{V}$
and seed set $\sde \subseteq \binom{V}{2}$ with $|\sde|\le r'$.
\end{inp}
\begin{outp} 
\item A set $T \subset V$ representing an approximation for \nusc\ problem.
\end{outp}
\begin{proc}
\itemsep=-0.5ex
\item $\sdn \gets \{ u \in V \mid 
\exists v: \{u,v\}\in \sde \} \subseteq V$. 
\item For each $f:\sdn \to \{0,1\}$,
\begin{enumerate}[(a)] 
\item Let $p^f: [n] \to V$ be an ordering of $V$ so that
$\langle \xvec_\sdn(f), \xvec_{p^f(1)} \rangle \le 
\ldots
\le \langle \xvec_\sdn(f), \xvec_{p^f(n)} \rangle$.
\item For each $i \in [n]$, 
\[T(f,i) \gets \left\{ p^f(1), p^f(2), \ldots, p^f(i) \right\}.\]
\end{enumerate}
\item $T \gets \argmin_{f:\sdn\to\{0,1\}, i\in [n]} \Phi_{T(f,i)}$.
\end{proc}
\end{program}
\begin{program}[h]
\caption{$\sde=$\textsc{select-seeds}($D,\xvec$): 
Seed selection in time $\poly(n)$. \label{alg:select-s}}
\begin{inp} 
\item $\xvec \in \lasserreii{2 r' + 2}{V}$ and $D: \binom{V}{2} \to \R_+$ as  
the demand graph.
\end{inp}
\begin{outp} 
\item $\sde \subseteq \binom{V}{2}$ with $|\sde|\le r'$ as a set of
  seed edges.
\end{outp}
\begin{proc}
\item Let $\widehat{\xmat} \gets \left[ \sqrt{D_{u,v}} \big( \xvec_u -
    \xvec_v\big)\right]_{\{u,v\}\in \binom{V}{2}}$.
\item Use the column selection algorithm from~\cite{gs11-svd} to
  choose $r'$-columns, $\sde \subseteq \binom{V}{2}$, of matrix
  $\widehat{\xmat}$ and return $\sde$.
\end{proc}
\end{program}
\begin{program}[t]
\caption{$T=$\textsc{Approximate-SC}($C,D,r'$): 
Main algorithm for approximating \nusc. Sparsity of the output
is bounded in~\Cref{cor:final-bnd}.
A na\"ive implementation will run in time $n^{O(r')}$. However 
this algorithm exactly fits into the local rounding framework 
introduced in~\cite{gs12-fast}, therefore we can use the faster solver
from~\cite{gs12-fast} to decrease the running time to $2^{O(r')} \poly(n)$.
\label{alg:vanilla-sc}}
\begin{inp} 
\item $C, D: \binom{V}{2} \to \R_+$ as 
the cost and demand
graphs, respectively.
\end{inp}
\begin{outp} 
\item A set $T \subset V$ representing an approximation for \nusc\
  problem.
\end{outp}
\begin{proc}
\item Compute a (near-)optimal solution, $\xvec$, to the following
  SDP:
\begin{equation}
\label{eq:sc-sdp-alg}
\begin{array}{rll}
  \min & \sum_{u<v} C_{u,v} \| \xvec_u - \xvec_v \|^2 \\
  \st & \sum_{u<v} D_{u,v} \| \xvec_u - \xvec_v \|^2 = 1, \\
  & \| \xvec_{\es}\|^2 > 0,\quad \xvec \in \lasserreii{2 r'+2}{V}.
\end{array}
\end{equation} 
\item $\sde \gets \textsc{select-seeds}(D,\xvec)$ (\Cref{alg:select-s}).
\item $T \gets \textsc{Round}(C,D,\xvec,\sde)$
  (\Cref{alg:rnd-from-s}). Return $T$.
\end{proc}
\end{program}
\label{sec:alg-and-analysis}
The complete algorithm is presented in~\Cref{alg:vanilla-sc}.  It is
based on rounding a certain $r'$-rounds of Lasserre Hierarchy
relaxation for the \nusc\ problem given positive integer $r'$:
\begin{align}
\label{eq:sc-sdp0}
\min \quad& \dfrac{\sum_{u<v} C_{u,v} \| \xvec_u - \xvec_v
  \|^2}{\sum_{u<v} D_{u,v}
  \| \xvec_u - \xvec_v \|^2} \\
\mathrm{st}\quad & \sum_{u<v} D_{u,v}
\| \xvec_u - \xvec_v \|^2 > 0, \notag\\
& \xvec \in \lasserreii{r'}{V},\quad\|\xvec_{\es}\|^2=1. \notag
\end{align} 
It is easy to see that~\cref{eq:sc-sdp0} is indeed a relaxation of
\nusc\ problem.  Even though it is not an SDP problem (it is
quasi-convex), there is an equivalent SDP formulation.
\begin{lemma} \label{lem:sdp-eq}
The following SDP is equivalent to~\cref{eq:sc-sdp0}:
\begin{align}
\label{eq:sc-sdp}
\min \quad&\sum_{u<v} C_{u,v} \| \wvec_u - \wvec_v \|^2 \\
\mathrm{st}\quad & \sum_{u<v} D_{u,v} \| \wvec_u - \wvec_v \|^2 = 1, \notag\\
&\| \wvec_{\es}\|^2 > 0,\quad \wvec \in \lasserreii{r'}{V}.\notag
\end{align} 
\end{lemma}
\begin{remark}
  The constraint $\|\wvec_{\es}\|^2 > 0$ in~\cref{eq:sc-sdp} is
  redundant, but we included it for the sake of clarity.
\end{remark}
\begin{proof} [Proof of~\Cref{lem:sdp-eq}] Given a feasible solution
  $\xvec$ for formulation \eqref{eq:sc-sdp0}, consider the following
  collection of vectors, $\wvec=[\wvec_T]_{T\in \binom{V}{\le r'}}$.
  For each $T\in \binom{V}{\le r'}$, we define $\wvec_T$ as $\wvec_T
  \triangleq \frac{1}{\sqrt{\sum_{u<v} D_{u,v} \| \xvec_u - \xvec_v
      \|^2}} \xvec_T$.
  It is easy to see that $\sum_{u<v} D_{u,v} \| \wvec_u - \wvec_v
  \|^2=1$ and objective values are equal.
  Finally $\|\wvec_{\es}\|^2 = \frac{1} {\sum_{u<v} D_{u,v} \| \xvec_u
    - \xvec_v \|^2} > 0$ since $0<\sum_{u<v} D_{u,v} \| \xvec_u -
  \xvec_v \|^2 < +\infty$.

  For the other direction of equivalence, suppose $\wvec$ is a
  feasible solution of~\cref{eq:sc-sdp}. For each $T\in \binom{V}{\le
    r'}, f\in \{0,1\}^T$, let $\xvec_T(f) \gets
  \frac{1}{\|\wvec_{\es}\|} \wvec_T(f)$. It is easy to see that the
  objective values are equal. The rest of the proof for $\xvec$ being
  a feasible solution of~\cref{eq:sc-sdp0} follows similarly to the
  previous direction.
\end{proof}
\begin{remark}
  The main components of our rounding,
  \Cref{alg:rnd-from-s,alg:select-s}, are scale invariant; thus the
  formulation given in~\cref{eq:sc-sdp} is sufficient for rounding
  purposes. But we chose to first present~\cref{eq:sc-sdp0} as it is
  more intuitive.
\end{remark}
\ngap
\subsection{Intuition Behind Our Rounding}
\label{sec:int}
For an intuition behind our rounding procedure, presented
in~\Cref{alg:rnd-from-s}, we start with a simple randomized rounding
procedure, which is based on the seed based propagation framework
from~\cite{gs11-qip}. In \cite{gs11-qip}, the different vertices were
rounded independently according to their marginal distribution
(conditioned on a partial assignment to the seed set), whereas here we
do correlated threshold rounding. The algorithm can be easily
derandomized as described at the end of this subsection. In the below
description, the details of the seed selection procedure will be
skipped and deferred to Section \ref{sec:seed-selection}. For now we
develop the algorithm and analyze it assuming some fixed choice of
seed edges.  (We will use the analysis as a guide to make a prudent
choice of seed edges.)

\vspace{0.1in}
\paragraph{Randomized rounding algorithm.} 
Consider the following procedure.
On input $\xvec \in \lasserreii{2 r' +2}{V}$:
\begin{enumerate}
\item Choose a set of $r'$-edges from the demand graph, say $\sde
  \subseteq \binom{V}{2}$ ({\em seed edges}).
\item Let $\sdn$ be the set of their endpoints, $\sdn \gets \{ u \in V
  \mid \mbox{exists $v$ such that $\{u,v\}\in \sde$} \} \subseteq V$.
\item Observe that $|\sdn| \le 2 r'$, hence the values
  $\big\|\xvec_\sdn(f)\big\|^2$ define a probability distribution over
  all labelings of $\sdn$, $f:\sdn \to\{0,1\}$ .  So sample a labeling
  for $\sdn$, $f:\sdn\to\{0,1\}$, with probability
  $\|\xvec_\sdn(f)\|^2$.
\item Choose a threshold $\tau \in [0,1]$ uniformly at random and
  output the following set:
\[
T(f,\tau) \triangleq \left\{ u \in V \bigg| \frac{\langle
    \xvec_\sdn(f), \xvec_u \rangle}{\|\xvec_\sdn(f)\|^2 } \ge \tau
\right\}.
\]
\end{enumerate}
In order for this procedure to make sense, the range of $
\frac{\langle \xvec_\sdn(f), \xvec_u \rangle}{\|\xvec_\sdn(f)\|^2 }$
should be similar to $\tau$'s range. In the following claim, we prove
this.
\begin{claim} \label{clm:tau-ulb}
Provided that $\xvec_\sdn(f) \neq 0$, we have:
\begin{inparaenum}[(i)]
\item \label{it:353-712-1} $0 \le \frac{\langle \xvec_\sdn(f), \xvec_u
    \rangle}{\|\xvec_\sdn(f)\|^2 } \le 1$ for any $u\in V$,
\item \label{it:353-712-2} $\frac{\left|\langle \xvec_\sdn(f), \xvec_u
      - \xvec_v \rangle\right|}{\|\xvec_\sdn(f)\|^2 } \le 1$ for any
  pair $u,v \in V$,
\item \label{it:353-712-3} $\frac{\langle \xvec_\sdn(f), \xvec_u
    \rangle}{\|\xvec_\sdn(f)\|^2} = f(u)$ for any $u \in \sdn$.
\end{inparaenum}
\end{claim}
\begin{proof}[of \eqref{it:353-712-1} and \eqref{it:353-712-2}]
  We will only prove \eqref{it:353-712-1}, from which
  \eqref{it:353-712-2} follows immediately.  The lower bound follows
  from $\langle \xvec_\sdn(f), \xvec_u \rangle = \|
  \xvec_{\sdn\cup\{u\}}(f\circ 1) \|^2 \ge 0$. For the upper bound, we
  have:
\begin{align*}
  \langle \xvec_\sdn(f), \xvec_\sdn(f) - \xvec_u \rangle
  =&  \|\xvec_\sdn(f)\|^2 - \langle \xvec_\sdn(f), \xvec_u \rangle \\
  =  & \langle \xvec_\sdn(f), \xvec_{\es}\rangle  
  - \langle \xvec_\sdn(f),\xvec_u \rangle \\
  =&  \langle \xvec_\sdn(f),\xvec_u(0) \rangle \\
  = & \| \xvec_{\sdn\cup\{u\}}(f\circ 0) \|^2\ge 0. \tag*{\qedhere}
\end{align*} 
\end{proof}
\begin{proof}[of \eqref{it:353-712-3}]
  Follows from the fact that $\langle \xvec_\sdn(f),\xvec_u(f(u))
  \rangle = \|\xvec_\sdn(f)\|^2$.
\end{proof}
%
Let's calculate the probability of separating two vertices 
by this procedure.
\begin{claim}\label{clm:sep-prob}
  $\expct{f,\tau}{ \left|\ind{T(f,\tau)}(u)-
      \ind{T(f,\tau)}(v)\right|} = \sum_f \left|\langle
    \xvec_{\sdn}(f), \xvec_u - \xvec_v\rangle\right|.  $
\end{claim}
\begin{proof}
  For fixed $f$, by~\Cref{clm:tau-ulb} the probability of separating
  $u$ and $v$ is equal to $ \frac{\left|\langle \xvec_{\sdn}(f),
      \xvec_u - \xvec_v\rangle\right|}{\|\xvec_\sdn(f)\|^2}$.
  Taking expectation over $f$:
\begin{multline*}
  \expct{f,\tau}{ \left|\ind{T(f,\tau)}(u)- \ind{T(f,\tau)}(v)\right|} \hfill\\
  \quad\quad\quad\quad= \sum_f \|\xvec_\sdn(f) \|^2
  \frac{\left|\langle \xvec_{\sdn}(f), \xvec_u -
      \xvec_v\rangle\right|}
  {\|\xvec_\sdn(f)\|^2} \hfill\\
  \quad\quad\quad\quad= \sum_f \left|\langle \xvec_{\sdn}(f), \xvec_u
    - \xvec_v\rangle\right|. \hfill \qedhere
\end{multline*}
\end{proof}

\paragraph{Derandomization.} 
For any fixed $f:\sdn\to\{0,1\}$, there are at most $n$ different
$T(f,\tau)$'s.  Hence instead of choosing $f:\sdn\to \{0,1\}$ and
$\tau\in [0,1]$ randomly, we can perform an exhaustive search over all
possible such sets and output the one with minimum sparsity.  Since
there are at most $n 2^{O(r')}$ many unique $T(f,\tau)$'s, the
exhaustive search can easily be implemented in time $\poly(n)
2^{O(r')}$.  The rounding procedure along with this modification is
presented in~\Cref{alg:rnd-from-s}.
%
%
\ngap
\subsection{Seed Based $\ell_1$-embedding}
Toward analyzing the rounding algorithm of the previous section, we
now define an embedding of the vertices into $\ell_1$, based on the
Lasserre solution (and the chosen seed edges).
\begin{definition}[Seed Based Embedding]
  \label{def:our-embedding}
  Given $\xvec \in \lasserreii{2 r'+2}{V}$ and $\sde \subseteq
  \binom{V}{2}$ with $|\sde|\le r'$, let $\sdn$ be the endpoints of
  edges in $\sde$ so that $\sde \subseteq \binom{\sdn}{2}$. Then we
  define the seed based embedding of $\xvec$ as the following
  collection of vectors.
  For each $u\in V$, $\yvec^\sde_u \in \R^{\{0,1\}^{\sdn}}$ is given
  by
  \[ \yvec^\sde_u \triangleq \bigg[ \langle \xvec_{\sdn}(f),
  \xvec_u\rangle \bigg]_{f:\sdn\to\{0,1\}} \ . \]
\end{definition}
Observe that $\|\yvec^\sde_u - \yvec^\sde_v\|_1$ is equal to the
probability that $u$ and $v$ are separated as shown
in~\Cref{clm:sep-prob}.

It is well known that once we have an $\ell_1$-embedding, we can get a
cut with similar sparsity by choosing the best threshold cut along
each coordinate and this is exactly what we do
in~\Cref{alg:rnd-from-s}.
%
%
%
The following lemma is well-known but for the sake of completeness we
provide a proof.
\begin{lemma}[\cite{llr}]
  \label{lem:l1-to-cut}
  Given a set of vertices $V$, a collection of vectors $\big[
  \yvec_{u} \in \R^{\Upsilon}\big]_{u \in V}$ representing an
  embedding of $V$, the following holds.  For any $C, D: \binom{V}{2}
  \to \R_+$ being the edge weights of graphs $G$ and $H$,
  respectively:
  \begin{equation}
    \min_{\substack{ f \in \Upsilon, \\ \tau \in \R} }
    \Phi_{T(f,\tau)} \le 
    \frac{ \sum_{u<v} C_{u,v} \left\| \yvec_u - \yvec_v \right\|_1 }
    { \sum_{u<v}D_{u,v} \left\| \yvec_u - \yvec_v \right\|_1 }.
    \label{eq:tcut-l1}
  \end{equation} 
  Here $T(f, \tau)\triangleq \left\{ u \in V\big|\ \yvec_u(f) \ge \tau
  \right\}$ represents the threshold cut along coordinate $f\in
  \Upsilon$.
\end{lemma}
\begin{proof}
  For any $f\in \Upsilon$, let $\delta_f \triangleq \max_{a,b} |
  \yvec_a(f)-\yvec_b(f) | = \max_b \yvec_b(f) - \min_{a} \yvec_a(f)$
  and $\Delta \triangleq \sum_f \delta_f$.
  Consider the following randomized process.  Choose $f \in \Upsilon$
  with probability proportional to $\delta_f$ and then sample a
  threshold $\tau \in_u [\min_{a} \yvec_a(f), \max_b \yvec_b(f)]$.
  Then:
  \begin{align}
    \label{eq:pwq115}
    \expct{f,\tau}{ \left|\ind{T(f,\tau)}(u) -
        \ind{T(f,\tau)}(v)\right| }
    = & \sum_{f\in \Upsilon} \frac{\delta_f}{\Delta} 
    \frac{|\yvec_u(f)-\yvec_v(f)|}{\delta_f}\\
    =& \frac{1}{\Delta} \left\| \yvec_u - \yvec_v\right\|_1. \notag
  \end{align}
  We have
  \begin{align*}
    \min_{\substack{ f \in \Upsilon, \\ \tau \in \R} }
    \Phi_{T(f,\tau)} =& \min_{\substack{ f \in \Upsilon, \\ \tau \in
        \R}} \frac{\sum_{u<v} C_{u,v} \left|\ind{T(f,\tau)}(u) -
        \ind{T(f,\tau)}(v)\right| }{
      \sum_{u<v} D_{u,v} \left|\ind{T(f,\tau)}(u) - \ind{T(f,\tau)}(v)\right|}  \\
    \le & \frac{ \expct{f,\tau}{ \sum_{u<v} C_{u,v}
        \left|\ind{T(f,\tau)}(u) - \ind{T(f,\tau)}(v)\right| } }{
      \expct{f,\tau}{ \sum_{u<v} D_{u,v} \left|\ind{T(f,\tau)}(u) -
          \ind{T(f,\tau)}(v)\right|
      }} \\
    = & \frac{ \sum_{u<v} C_{u,v} \left\| \yvec_u - \yvec_v \right\|_1
    } { \sum_{u<v}D_{u,v} \left\| \yvec_u - \yvec_v \right\|_1 } \quad
    \mbox{ (using \cref{eq:pwq115})} 
  \end{align*}
  thus proving the lemma.
\end{proof}			 

In the rest of this section, we will upper bound the right hand side
of ~\cref{eq:tcut-l1} for our embedding
from~\Cref{def:our-embedding}. To this end, we now obtain upper and
lower bounds on the $\ell_1$-distance $\| \yvec^\sde_u - \yvec^\sde_v
\|_1$ in terms of the SDP vectors.

\begin{claim} \label{clm:l1-ub} $\| \yvec^\sde_u - \yvec^\sde_v \|_1
  \le \left\| \xvec_u - \xvec_v \right\|^2$.
\end{claim}
\begin{proof}
  Since $\xvec \in \lasserreii{2 r'+2}{V}$, we can express $\xvec_u$
  and $\xvec_v$ as $\xvec_u = \xvec_{u,v}(10)+\xvec_{u,v}(11)$ and
  $\xvec_v = \xvec_{u,v}(01)+\xvec_{u,v}(11)$ respectively.  Thus
  $\xvec_u - \xvec_v = \xvec_{u,v}(10)- \xvec_{u,v}(01)$.
  Now the following identity follows easily\footnote{ Intuitively, it
    corresponds to the following.  The ``probability'' of $u$ and $v$
    are separated is equal to the probability of $u$ and $v$ being
    labeled with $1$ and $0$ or $0$ and $1$.  }:
  \begin{align}
    \label{eq:7112-306}
    \|\xvec_u - \xvec_v \|^2 = &
    \|\xvec_{u,v}(10)\|^2 + \|\xvec_{u,v}(01)\|^2 \\
    & - 2 \underbrace{\langle \xvec_{u,v}(10),\xvec_{u,v}(01)
      \rangle }_{=0} \notag \\
    =& \|\xvec_{u,v}(10)\|^2 + \|\xvec_{u,v}(01)\|^2. \notag
  \end{align}
  Therefore:
  \begin{align}
    \left|\yvec^\sde_u(f) - \yvec^\sde_v(f) \right|
    = & \left| \langle \xvec_\sdn(f), \xvec_u - \xvec_v \rangle
    \right|
    \notag \\
    = & \left| \langle \xvec_\sdn(f), \xvec_{u,v}(10)
      - \xvec_{u,v}(01) \rangle \right| \notag \\
    \le  & \left|  \langle \xvec_\sdn(f), \xvec_{u,v}(10) \rangle
    \right| 
    \notag\\
    & + \left| \langle \xvec_\sdn(f), \xvec_{u,v}(01) \rangle \right|
    \notag \intertext{For any $g:\{u,v\}\to\{0,1\}$, $\langle
      \xvec_\sdn(f), \xvec_{u,v}(g) \rangle = \big\| \xvec_{\sdn\cup
        \{u,v\}} (f\circ g)\big\|^2 \ge 0$.
      Thus:} = & \langle \xvec_\sdn(f), \xvec_{u,v}(10) \rangle +
    \langle \xvec_\sdn(f), \xvec_{u,v}(01) \rangle
    . \label{eq:42310112} \intertext{Summing~\cref{eq:42310112} over
      $f$ and using the fact that $\xvec_{\es} = \sum_f
      \xvec_{\sdn}(f)$:} \| \yvec^\sde_u - \yvec^\sde_v \|_1 \le &
    \Big\langle \sum_f \xvec_\sdn(f), \xvec_{u,v}(10)
    + \xvec_{u,v}(01) \Big\rangle \notag\\
    =& \langle \xvec_{\es}, \xvec_{u,v}(10) + \xvec_{u,v}(01) \rangle
    \notag \\
    = & \| \xvec_{u,v}(10)\|^2 + \|\xvec_{u,v}(01) \|^2 \notag\\
    = &\| \xvec_u-\xvec_v\|^2 \text{ by~\cref{eq:7112-306}.}
    \tag*{\qedhere}
  \end{align}
\end{proof}
\begin{claim}\label{clm:l1-lb}
  $\| \yvec^\sde_u - \yvec^\sde_v \|_1 \ge \sum_{f:
    \xvec_{\sdn}(f)\neq 0} \langle \unit{\xvec_\sdn(f)}, \xvec_u -
  \xvec_v \rangle^2$ where $\unit{\xvec_\sdn(f)} \triangleq
  \frac{\xvec_\sdn(f)}{\|\xvec_\sdn(f)\|}$ is the unit vector for
  $\xvec_{\sdn}(f)$.
\end{claim}
\begin{proof}
For any $f: \xvec_{\sdn}(f)\neq 0$, 
by~\Cref{clm:tau-ulb}, $0\le \frac { \left|\langle {\xvec_\sdn(f)},
    \xvec_u - \xvec_v \rangle \right|} { \|\xvec_{\sdn}(f)\|^2 } \le
1$ thus
$\frac { \left|\langle {\xvec_\sdn(f)}, \xvec_u - \xvec_v \rangle
  \right|} { \|\xvec_{\sdn}(f)\|^2 } \ge \left( \frac { \langle
    {\xvec_\sdn(f)}, \xvec_u - \xvec_v \rangle } {
    \|\xvec_{\sdn}(f)\|^2 } \right)^2$.  
Multiplying both sides with
$\|\xvec_{\sdn}(f)\|^2>0$, we obtain
\[ \left|\langle {\xvec_\sdn(f)}, \xvec_u - \xvec_v \rangle \right|
\ge \frac{\langle {\xvec_\sdn(f)}, \xvec_u - \xvec_v \rangle^2}{
  \|\xvec_{\sdn}(f)\|^2} \ . \]
Summing over all $f: \xvec_{\sdn}(f)\neq 0$, we obtain the desired lower
bound,
\begin{align*}
  \|\yvec^\sde_u - \yvec^\sde_v\|_1 = &
  \sum_f \left|\langle {\xvec_\sdn(f)}, \xvec_u - \xvec_v \rangle \right| \\
  \ge& \sum_{f:\xvec_{\sdn}(f)\neq 0} \frac{\langle {\xvec_\sdn(f)},
    \xvec_u - \xvec_v \rangle^2}{ \|\xvec_{\sdn}(f)\|^2}.
  \tag*{\qedhere}
\end{align*}
\end{proof}
In its current form, our lower bound is not very useful as it involves
the {\em higher order} vectors ($\xvec_{\sdn}(f)$'s) from our
relaxation.  Unfortunately these vectors are very hard to reason
about: We do not have any direct handle on them.  Therefore our goal
is to relate this expression to some other expression that only
involves the vectors for edges ($\xvec_u-\xvec_v$'s). We first
introduce some notation.
\begin{notation}
  Let $ \Pi_\sdn \triangleq \sum_{f: \xvec_\sdn(f)\neq 0}
  \unit{\xvec_\sdn(f)} \cdot \unit{\xvec_\sdn(f)}^T$.
\end{notation}
We can rewrite the lower bound from~\Cref{clm:l1-lb} 
in terms of $\Pi_{\sdn}$ as follows:
\begin{align}
  \sum_f & \langle \unit{\xvec_\sdn(f)}, \xvec_u - \xvec_v \rangle^2 \notag \\
  = & (\xvec_u - \xvec_v)^T \sum_f \unit{\xvec_\sdn(f)} \cdot
  \unit{\xvec_\sdn(f)}^T
  (\xvec_u - \xvec_v) \notag \\
  =& (\xvec_u - \xvec_v)^T \Pi_{\sdn} (\xvec_u - \xvec_v).
\label{eq:601-7112}
\end{align}
As observed in~\cite{gs11-qip}, 
$\Pi_{\sdn}$ has a special structure --- it is a projection matrix
onto the span of vectors $\{\xvec_{\sdn}(f)\}_f$. 
\begin{proposition}\label{lem:pi-is-proj}
  $\Pi_\sdn^2 = \Pi_\sdn$, i.e. $\Pi_\sdn$ is a projection matrix onto
  the span of vectors in $\{\xvec_\sdn(f)\}$.
\end{proposition}
\begin{proof}
Observe that 
$\langle \unit{\xvec_\sdn(f)}, \unit{\xvec_\sdn(g)}\rangle  = \begin{dcases*}
1 & if $f=g$, \\
0 & else
\end{dcases*}$. 
Then we have:
\begin{align*}
\Pi_\sdn^2 
= & \sum_{f, g} \langle \unit{\xvec_\sdn(f)},
\unit{\xvec_\sdn(g)}\rangle
\unit{\xvec_\sdn(f)}\cdot  \unit{\xvec_\sdn(g)}^T\\
= & \sum_f \unit{\xvec_\sdn(f)}\cdot \unit{\xvec_\sdn(f)}^T =
\Pi_\sdn.  \tag*{\qedhere}
 \end{align*} 
\end{proof}
For each seed edge $\{u,v\} \in \sde$,
$\xvec_u - \xvec_v \in \spn\left\{ \xvec_\sdn(f)\right\}$.
This means we can lower bound the matrix $\Pi_{\sdn}$ in terms of the
projection matrix onto the span of vectors corresponding to seed
edges!
%
%
\begin{notation}
\label{not:proj-mtx}
Let $P_{\sde}$ be the projection matrix onto the span of $\{\xvec_u -
\xvec_v \}_{\{u,v\}\in \sde}$. Similarly let $P_{\sde}^{\perp}$ be
projection matrix onto the orthogonal complement of $\{\xvec_u -
\xvec_v \}_{\{u,v\}\in \sde}$, i.e., $P_{\sde}^{\perp} = I -
P_{\sde}$. Here $I$ is the identity matrix of appropriate dimension.
\end{notation}
The final ingredient in our embedding is to lower bound the $\ell_1$
distance.
\begin{lemma} \label{lem:l1-lb} $ \| \yvec^\sde_u - \yvec^\sde_v \|_1
  \ge \left\| P_{\sde} (\xvec_u - \xvec_v)\right\|^2 = \left\| \xvec_u
    - \xvec_v\right\|^2 - \left\| P_{\sde}^\perp (\xvec_u -
    \xvec_v)\right\|^2 $.
\end{lemma}
\begin{proof}
  From~\Cref{clm:l1-lb} and~\cref{eq:601-7112} we see that $
  \|\yvec^\sde_u - \yvec^\sde_v\|_1 \ge (\xvec_u - \xvec_v)^T
  \Pi_{\sdn} (\xvec_u - \xvec_v).  $ For any $u\in \sdn$, $\xvec_u =
  \sum_{f:f(u)=1} \xvec_\sdn(f)$ hence $\xvec_u \in
  \spn\{\xvec_\sdn(f)\}$. In particular, for any pair $u,v\in \sdn$:
  $\xvec_u - \xvec_v \in \spn\{\xvec_\sdn(f)\}$, which means:
  \begin{align*}
    \spn \left\{ \xvec_u - \xvec_v \right\}_{\{u,v\} \in \sde}
    \subseteq & \spn \left\{ \xvec_u - \xvec_v  \right\}_{u,v \in \sdn} \\
    \subseteq & \spn \left\{ \xvec_{\sdn}(f) \right\} \\
    \implies \Pi_{\sdn}\succeq P_{\sde} = & P_{\sde}^2.
  \end{align*}
  Consequently, $(\xvec_u - \xvec_v)^T \Pi_{\sdn} (\xvec_u - \xvec_v)
  \ge (\xvec_u - \xvec_v)^T P_{\sde}^2 (\xvec_u - \xvec_v) = \left\|
    P_{\sde} (\xvec_u - \xvec_v) \right\|^2$. \qedhere
\end{proof}
We wrap up this section with the following theorem which bounds the
approximation factor of our \Cref{alg:rnd-from-s} in terms of the SDP
vectors $x_u$ and the projection matrix $P_{\sde}^\perp$ corresponding
to the seed edges $\sde$.
\newcommand{\scsdp}{\Phi^{\mathrm{SDP}}}
\begin{theorem}\label{thm:rnd-from-s}
  Given $\xvec \in \lasserreii{r'}{V}$ and a set of seed edges $\sde
  \subseteq \binom{V}{2}$ with projection matrices $P_{\sde},
  P_{\sde}^{\perp}$ as in~\Cref{not:proj-mtx}; let $T\subset V$ be the
  set returned by \Cref{alg:rnd-from-s}.
  Then the sparsity $\Phi_T$ of $T$ is bounded by:
  \begin{equation}
    \Phi_T \le 
    \scsdp \cdot \left(
      1 - \frac{\sum_{u<v} D_{u,v} \|P_{\sde}^{\perp} (\xvec_u- \xvec_v)\|^2}
      {\sum_{u<v} D_{u,v} \|\xvec_u - \xvec_v\|^2}
    \right)^{-1}
    \label{eq:rnd-bnd}
  \end{equation} where
  $\scsdp \triangleq 
   \frac{\sum_{u<v} C_{u,v} \|\xvec_u - \xvec_v\|^2 } 
        {\sum_{u<v} D_{u,v} \|\xvec_u - \xvec_v\|^2 }.$
\end{theorem}
\begin{proof}
  $\Phi_T \le \frac{\sum_{u<v} C_{u,v} \|\yvec^\sde_u -
    \yvec^\sde_v\|_1} {\sum_{u<v} D_{u,v} \|\yvec^\sde_u -
    \yvec^\sde_v\|_1}$ follows from \Cref{lem:l1-to-cut}.
  \Cref{clm:l1-ub} and \Cref{lem:l1-lb} together imply:
\begin{align*}
  \frac{\sum_{u<v} C_{u,v} \|\yvec^\sde_u - \yvec^\sde_v\|_1}
  {\sum_{u<v} D_{u,v} \|\yvec^\sde_u - \yvec^\sde_v\|_1} \le&
  \frac{\sum_{u<v} C_{u,v} \|\xvec_u - \xvec_v\|^2} { \sum_{u<v}
    D_{u,v} \|\xvec_u - \xvec_v\|^2 - \sum_{u<v} D_{u,v}
    \|P_{\sde}^{\perp} (\xvec_u - \xvec_v)\|^2
  } \\
  = & \scsdp \left( 1 - \frac{\sum_{u<v} D_{u,v} \|P_{\sde}^{\perp}
      (\xvec_u- \xvec_v)\|^2} {\sum_{u<v} D_{u,v} \|\xvec_u -
      \xvec_v\|^2} \right)^{-1} \tag*{\qedhere}.
\end{align*}
\end{proof}
%
\ngap
\subsection{Choosing Seed Edges}
\label{sec:seed-selection}

We now turn to the main missing piece in our algorithm and its
analysis: how to make a good choice of the seed edges $\sde$, and how
to relate the guarantee of \cref{eq:rnd-bnd} for that choice of $\sde$
to the generalized eigenvalues between the Laplacians of the cost and
demand graphs.

\begin{notation} Given $\xvec=[\xvec_T \in \R^\Upsilon]$ and
  $D:\binom{V}{2} \to \R_+$, let 
  $\widehat{\xmat} \in \R^{\Upsilon, \binom{V}{2}}$ be the following
  matrix whose columns are associated with vertex pairs:
  $
  \widehat{\xmat} \triangleq \left[ \sqrt{D_{u,v}} (\xvec_u-\xvec_v)
  \right]_{\{u,v\}\in \binom{V}{2}}$.
  \label{not:xmat}
\end{notation}
Observe that $\left\| \widehat{\xmat} \right\|^2_F = \sum_{u<v}
D_{u,v} \left\| \xvec_u - \xvec_v\right\|^2$.
Since $\sde \subseteq \binom{V}{2}$, the matrix
$\widehat{\xmat}_{\sde}$, consisting of columns of $\xmat$ indexed by
$\sde$, is well defined.  Moreover there is a strong connection
between $\widehat{\xmat}_{\sde}^{\proj}$ and $P_{\sde}$, which we
formalize next:
\begin{claim}
  $P_\sde \succeq {(\widehat{\xmat}_{\sde})}^{\proj}$. Furthermore if
  $\sde \subseteq \supp(D)$ then $P_\sde =
  {(\widehat{\xmat}_{\sde})}^{\proj}$.
\end{claim}
\begin{proof}
  Recall that $\sde\subseteq\binom{V}{2}$ and $P_{\sde}$ represents
  $\spn\{\xvec_u-\xvec_v\}_{\{u,v\}\in \sde}$, which contains every
  column of $\widehat{\xmat}_{\sde} = \left[ \sqrt{D_{u,v}} \left(
      \xvec_u-\xvec_v\right) \right]_{\{u,v\}\in \sde}$.
\end{proof} 

After substituting~\Cref{not:xmat}, the approximation factor in
\Cref{thm:rnd-from-s} becomes
\[ \left(1 - \frac{\|{(\widehat{\xmat}_{\sde})^{\perp}}\widehat{\xmat}
    \|^2_F}{\|\widehat{\xmat}\|^2_F}\right)^{-1} \ .
\]
One way to think about
$\|{(\widehat{\xmat}_{\sde})^{\perp}}\widehat{\xmat} \|^2_F$ is in
terms of column based matrix reconstruction. If we were to express
each column of $\widehat{\xmat}$ as a linear combination of only
$r$-columns of $\widehat{\xmat}$, what is the minimum reconstruction
error (in terms of Frobenius norm) we can achieve?
Without the restriction of choosing only columns, this question
becomes easy to answer: the best error is achieved by the top $r$
singular vectors of $\widehat{\xmat}$ and equals the sum of all but
largest $r$ eigenvalues of the Gram matrix, $\widehat{\xmat}^T
\widehat{\xmat}$. For convenience, we record this in
\Cref{lem:col-res-lb} below.
\begin{notation}
  Let $\sigma_1 \ge \sigma_2 \ge \ldots \ge \sigma_m \ge 0$ be the
  eigenvalues of $\widehat{\xmat}^T \widehat{\xmat}$ in descending
  order.
\end{notation}
\begin{claim}
  \label{lem:col-res-lb}
  For any seed set $\sde\subseteq \binom{V}{2}$ with $|\sde|=r$,
  \vnote{Changed $|\sde|=r-1$ to $r$.}
  $ \|{(\widehat{\xmat}_{\sde})^{\perp}}\widehat{\xmat} \|^2_F \ge
  \sum_{j\ge r+1} \sigma_j $.
\end{claim}
%
\aknote{Happily removed the proof :).}
In~\cite{gs11-svd}, it was shown that choosing $\sim \frac{r}{\eps}$
many columns suffice to decrease the error within a $(1+\eps)$-factor
of this lower bound and this is essentially the best possible up to
low order terms.
\begin{theorem}[\cite{gs11-svd}] \label{thm:choose-s} For any positive
  integer $r$ and positive real $\eps$, there exists $r' =
  \left(\frac{r}{\eps}+r-1\right)$ columns of $\widehat{\xmat}$,
  $\sde$, such that
	\[
        \left\| \big(\widehat{\xmat}_{\sde}\big)^{\perp}
          \widehat{\xmat}\right\|^2_F \le \left( 1 + \eps \right)
        \sum_{j\ge r+1} \sigma_j \ .
        \]
  Furthermore there exists an algorithm to find such $\sde$ in
  time $\poly(n)$ (recall $\widehat{X}$ has $\binom{n}{2} =
  O(n^2)$ columns).
\end{theorem}
Our seed selection procedure is presented in~\Cref{alg:select-s}.  We
bound $\sum_{j\ge r+1} \sigma_j$ in~\Cref{lem:sigma-bnd}. The main
approximation algorithm combining \Cref{alg:rnd-from-s,alg:select-s}
is presented in~\Cref{alg:vanilla-sc} with its analysis
in~\Cref{cor:final-bnd}.
\begin{theorem} \label{lem:sigma-bnd} Let $0\le \lambda_1\le \lambda_2
  \le \ldots \le \lambda_m$ be the generalized eigenvalues of
  Laplacian matrices for the cost and demand graphs.  Then for
  $\widehat{\xmat}$ being the matrix defined in~\Cref{not:xmat}, the
  following bound holds:
\[
\frac { \sum_{j\ge r+1} \sigma_j }{ \big\| \widehat{\xmat} \big\|^2_F
} \le \frac{\scsdp}{\lambda_{r+1}}.
\]
\end{theorem}
Before proving~\cref{lem:sigma-bnd}, we will begin with stating a
simple lower bound on the trace of matrix products in terms of the
spectra.
\begin{lemma}\label{thm:birkhoff} Given a symmetric matrix $X$ and
  positive semidefinite matrix $Y$:
\[
\tr(X Y) \ge \sum_j \sigma_j(X) \lambda_j(Y)
\] where $\sigma_j(X)$ and $\lambda_j(Y)$ denote the $j^{th}$ largest
eigenvalue of $X$ and the $j^{th}$ smallest eigenvalues of $Y$,
respectively.
\end{lemma} 
\begin{proof}
  von Neumann's Trace Inequality~\cite{hj-mat-book} states that $\tr(A
  B) \le \sum_j \sigma_j(A) \sigma_j(B)$ for any pair of symmetric
  matrices, $A$ and $B$. This allows us to lower bound $\tr(X Y)$ as
  follows, from which the claim follows immediately:
\begin{align*}
  \tr(X (-Y)) \le & \sum_i \sigma_i(X) \sigma_i(-Y) \\
  =& \sum_i \sigma_i(X) (-\lambda_i(Y)) = - \sum_i \sigma_i(X)
  \lambda_i(Y).  \qedhere
\end{align*} 
\end{proof}
\begin{proof}[Proof of \Cref{lem:sigma-bnd}]
  Since the claimed bound is scale independent, we may assume $\big\|
  \widehat{\xmat} \big\|^2_F=1$ without loss of generality.

  Throughout the proof, we will use the following matrices:
\begin{itemize}
\item $\xmat \triangleq \left[ \xvec_u \right]_{u\in V} \in \R^{\Upsilon, V}$,
\item $B_C \in \R^{\binom{V}{2}, V}$ is the following edge-node
  incidence matrix of the cost graph whose columns and rows are
  associated with vertices and edges, respectively. Its entry at
  column $c\in V$ and row $\{a,b\}\in \binom{V}{2}$ with $a<b$
  (assuming some consistent ordering of $V$) is given by:
\[
\left( B_C \right)_{\{a,b\}, c} \triangleq \sqrt{C_{u,v}}
		\begin{dcases*}
                  1 & if $c = a$, \\
                  -1 & if $c = b$, \\
                  0 & else.
		\end{dcases*}
\]
\item $B_D\in \R^{\binom{V}{2}, V}$ is defined similarly for the
  demand graph, with its rank being $R$.
\item Singular value decomposition of $B_D$ is given by $B_D = P
  \Lambda^{1/2} Q^T$ with $Q^T Q = P^T P = I_{R}$. Here $P$ is an
  orthonormal $\binom{n}{2}$-by-$R$ matrix, $\Lambda$ is an $R$-by-$R$
  positive diagonal matrix and $Q$ is an orthonormal $n$-by-$R$

  matrix. 
\item $L_C, L_D$ are the Laplacian matrices for cost and demand
  graphs, respectively.
\item $(L_D)^{\dagger}$ is the pseudo-inverse of $L_D$ so that
  $L_D^{\dagger} = Q \Lambda^{-1} Q^T$.
\item $Z\triangleq L_D^{\dagger} L_C L_D^{\dagger}$.
\end{itemize}
\vnote{Was the earlier argument, which argued about spectrum of $Z$ to
  capture generalized eigenvalues, wrong?}  \aknote{Nope. But I
  thought this proof is simpler.}
The following identities are trivial:
\[
\widehat{\xmat} = \xmat (B_D)^T;\quad L_C = B_C^T B_C;\quad L_D =
B_D^T B_D = Q \Lambda Q^T.
\] 
Moreover
\[
\left\| \xmat B_C^T \right\|^2_F = \tr\left( \xmat L_C \xmat^T\right)
= \scsdp \tr\left( \xmat L_D \xmat^T\right) = \scsdp
\] by our assumption that $\|\widehat{\xmat}\|_F^2 = \tr\left( \xmat
  L_D \xmat^T\right)=1$.
	
Our goal is to lower bound $\scsdp = \tr(X L_C X^T)$ by $\lambda_{r+1}
\sum_{j \ge r+1} \sigma_j$.  Our approach will be to use
\Cref{thm:birkhoff} to prove this, by identifying a suitable matrix
whose eigenvalues equal the generalized eigenvalues $\lambda_j$.
		
Since $(L_D)^{\Pi}$ is a projection matrix and $L_C \succeq 0$, we
have $L_C \succeq (L_D)^{\Pi} L_C (L_D)^{\Pi}$ .  Substituting the
identity $(L_D)^{\Pi} = L_D (L_D)^{\dagger}= (L_D)^{\dagger} L_D$ into
this lower bound, we have:
\begin{align}
  L_C \succeq & L_D L_D^{\dagger} L_C L_D^{\dagger} L_D  \notag \\
  = & (B_D)^T \Big[B_D \underbrace{L_D^{\dagger} L_C L_D^{\dagger}}_{=
    Z}
  (B_D)^T\Big] B_D. \notag\\
  \implies \scsdp =& \tr\left( \xmat L_C \xmat^T \right) \notag\\
  \ge & \tr\Big\{ \xmat (B_D)^T \left[B_D Z (B_D)^T\right] B_D
  \xmat^T \Big\} \notag \\
  = & \tr\Big\{ \widehat{\xmat} \left[B_D Z (B_D)^T\right]
  \widehat{\xmat}^T\Big\} .\label{eq:44610112}
\end{align}
Since the null space of $\widehat{X} = X (B_D)^T$ contains the null
space of $(B_D)^T$, the row span of $\widehat{X}$ is contained in the
span of $P$. Recall that $P$ is an orthonormal matrix, therefore $P
P^T$ is a {\em projection} matrix onto its column span. Consequently
$\widehat{\xmat} P P^T = \widehat{\xmat}$ and~\cref{eq:44610112} is
equal to: \vnote{It might be good to remind reader of relation between
  SVD and null space and justify $\widehat{\xmat} P P^T =
  \widehat{\xmat}$.}
\begin{equation}
  \tr\Big\{ \widehat{\xmat} P P^T \left[B_D Z (B_D)^T\right] 
  P P^T\widehat{\xmat}^T\Big\}.
\label{eq:44610113}
\end{equation}
If we substitute the lower bound from~\Cref{thm:birkhoff}
into~\cref{eq:44610113}, we see that
\begin{align*}
  \sum_j & \sigma_j (\widehat{\xmat} P P^T \widehat{\xmat}^T) 
  \lambda_j(P^T B_D Z B_D^T P)  \\
  & = \sum_j \sigma_j (\widehat{\xmat} \widehat{\xmat}^T)
  \lambda_j(P^T B_D Z B_D^T P),
\end{align*}
where we used $\sigma_j(M)$, $\lambda_j(N)$ to denote the $j^{th}$
largest and smallest eigenvalues of matrices $M$ and $N$,
respectively.  Observing that $\sigma_j(\widehat{X}
\widehat{X}^T)=\sigma_j(\widehat{X}^T \widehat{X}) = \sigma_j$, we
finally obtain \cref{eq:432}:
\begin{equation}
\label{eq:432}
\scsdp = \tr(X L_C X^T) \ge \lambda_{r+1}(P^T B_D Z B_D^T P) 
\sum_{j\ge r+1} \sigma_j.
\end{equation}
To complete the proof, we need to relate $\lambda_i(P^T B_D Z B_D^T
P)$ to $\lambda_i(L_C, L_D)$.
%
%
\begin{claim} \label{clm:50910112}
$\lambda_j(P^T B_D Z B_D^T P) = \lambda_j(L_C, L_D)$.
\end{claim}
\begin{proof} First observe that $P^T B_D Z B_D^T P = \Lambda^{1/2}
  Q^T L_D^{\dagger} L_C L_D^{\dagger} Q \Lambda^{1/2} = \Lambda^{-1/2}
  Q^T L_C Q \Lambda^{-1/2}$. Since $\Lambda$ is positive definite, any
  eigenvalue of $\Lambda^{-1/2} Q^T L_C Q \Lambda^{-1/2}$ is also a
  generalized eigenvalue of matrices $Q^T L_C Q$ and $\Lambda$. $Q$ is
  an orthonormal matrix, therefore the generalized spectrum does not
  change when we transform both matrices by $Q$, implying that
  $\lambda_i(Q^T L_C Q, \Lambda) = \lambda_i(Q Q^T L_C Q Q^T, Q
  \Lambda Q^T) = \lambda_i(Q Q^T L_C Q Q^T, L_D)$. Proof is complete
  by observing that any generalized eigenvector, $z$, lies in the span
  of $L_D$ so that $Q Q^T z = z$. Hence $\lambda_i(Q Q^T L_C Q Q^T,
  L_D) = \lambda_i(L_C, L_D)$.
\end{proof}
Proof is complete by combining \Cref{clm:50910112} and \cref{eq:432}.
\end{proof}
We put everything together in our main result below.
\begin{theorem}
\label{cor:final-bnd}
Given $C, D: \binom{V}{2}\to \R_+$ representing cost and demand
graphs, let $ 0\le \lambda_1\le \lambda_2 \le \ldots \le \lambda_m $
be their generalized eigenvalues in ascending order.
For any positive integer $r$ and real $\eps>0$, on input $C$, $D$ and
$r' \triangleq \frac{r}{\eps}+r-1$, \Cref{alg:vanilla-sc} outputs a
subset $T\subset V$ whose sparsity $\Phi_T$ (w.r.t cost graph $C$ and
demand graph $D$) is bounded by:
\[
\Phi_T \le \Phi^\ast \left( 1 - (1+\eps)
  \frac{\Phi^\ast}{\lambda_{r+1}}\right)^{-1} \ \text{if}\
(1+\eps) \frac{\Phi^\ast}{\lambda_{r+1}} < 1.
\]
Furthermore using the SDP solver from~\cite{gs12-fast}, the running
time can be decreased to $2^{O(r')} \poly(n)$.  \vnote{Where does the
  ``provided" constraint come from? It is not transparent from the
  proof.}  \aknote{Necessary otherwise the quantity is negative, which
  means the upper bound becomes lower bound when we take the
  reciprocal.}
\end{theorem}
\begin{proof} Whenever $(1+\eps) \frac{\Phi^\ast}{\lambda_{r+1}} < 1$,
  the quantity $1 - (1+\eps) \frac{\Phi^\ast}{\lambda_{r+1}}$ is
  positive which means
\[\left(1 - \frac{\|{(\widehat{\xmat}_{\sde})^{\perp}}\widehat{\xmat}
    \|^2_F} {\|\widehat{\xmat}\|^2_F}\right)^{-1} \le \left(1 -
  (1+\eps) \frac{\Phi^\ast}{\lambda_{r+1}} \right)^{-1}
\] by \Cref{thm:choose-s,lem:sigma-bnd}. Substituting the bound
from~\Cref{thm:rnd-from-s} completes the proof.
\end{proof}
There are two interesting regimes in~\Cref{cor:final-bnd}, which we
highlight in~\Cref{cor:final-bnd-regimes}.
\begin{corollary}
\label{cor:final-bnd-regimes}
Given $C, D: \binom{V}{2}\to \R_+$ representing cost and demand
graphs, positive real $0<\delta \le \frac{1}{2}$ and positive integer
$r$, there exists an algorithm which outputs a subset $T\subset V$, in
time $2^{O(r/\delta)} n^{O(1)}$, such that:
\begin{itemize}
\item (Near Optimal) If $\Phi^\ast < \frac{1}{2} \delta \cdot
  \lambda_{r}$ then $\Phi_T \le \Phi^\ast (1+\delta)$;
\item (Constant Factor) Otherwise if $\Phi^\ast < (1-2 \delta)
  \lambda_{r}$ then $\Phi_T \le \frac{\Phi^\ast}{\delta}$.
\end{itemize}
\end{corollary}
%
\section{Using Subspace Enumeration for \usc}
\label{apx:rse}
\vnote{Should we acknowledge the reviewer here again?} \aknote{Done.}
Throughout this section, we will assume that the cost graph with
weights $C:\binom{V}{2}\to \R_+$ and $\sum_u C_{u,v} = 1$ for any
$v$. 
Since $G$ is regular, definitions of uniform sparsest cut / normalized
cut and edge expansion / conductance coincide. Thus we will focus only
on \usc\ which we denote by $\phi^\ast$.

\def\uscopt{\phi^\ast} The following theorem is adapted
from~\cite{al08} for our setting:
\begin{theorem}[Cut Improvement, see~\cite{al08}]
  For any $x^\ast\in \{0,1\}^{V}$, given $x \in \{0,1\}^{V}$
  satisfying
  \[0<{\|x\|_1}{} \le \frac{n}{2} \text{ and }
  \frac{\langle x, x^\ast\rangle}{\|x^\ast\|_1} > \frac{ \|x\|_1}{n}
\] 
in polynomial time one can find $y\in \{0,1\}^{V}$ whose edge
expansion is within a factor \[
\le \frac{ 1 - \|x\|_1/n }{ \langle x, x^\ast\rangle/\|x^\ast\|_1 -
  \|x\|_1/n }\] of $x^\ast$'s edge expansion.
\end{theorem}
\vnote{What happened to the proof you wrote for this?}  \aknote{I only
  had proof of corollary, which is commented out. Do you want to put
  it back?}  The following is adapted from~\cite{ABS}:
\begin{theorem}[Eigenspace Enumeration~\cite{ABS}]
  In time $2^{O(r)} n^{O(1)}$, there exists an algorithm which outputs
  a set $X\subseteq \{0,1\}^{V}$ that contains some $x\in X$ with
  following property.  There exists $x^\ast\in \{0,1\}^{V}$ with:
\[
\frac{\|x - x^\ast\|_1}{ \|x^\ast\|_1 } \le \frac{8}{\lambda_r}
\phi^\ast. 
\]
\end{theorem}
Combining these two, one can obtain the following\footnote{We thank an
  anonymous reviewer for pointing out this combination for uniform
  sparsest cut.}:
\begin{corollary} %
  For any positive integer $r$, if $r^{th}$ smallest eigenvalue of
  Laplacian matrix for cost graph satisfies $\lambda_r > 8 \phi^\ast$
  where $\phi^\ast$ is the \usc\ value, then in time $n^{O(1)}
  2^{O(r)}$ one can find $y\in \{0,1\}^{V}$ whose uniform sparsity is
  bounded by:
  \[
  \frac{2 \phi^\ast}{ 1 - 8 \frac{\phi^\ast}{\lambda_r}}.
  \]
\end{corollary}
\ngap
\section*{Acknowledgments}
We thank anonymous reviewers for their helpful comments on earlier
drafts of this paper as well as pointing out the combination of
\cite{al08} with \cite{ABS} to approximate \usc.
%
\bibliographystyle{abbrv}
\bibliography{../hdr-bib/references}
\end{document}